\documentclass{article}

\usepackage{PRIMEarxiv}

\usepackage[utf8]{inputenc} 
\usepackage[T1]{fontenc}    
\usepackage{hyperref}       
\usepackage{url}            
\usepackage{booktabs}       
\usepackage{amsfonts}       
\usepackage{nicefrac}       
\usepackage{microtype}      
\usepackage{lipsum}
\usepackage{fancyhdr}       
\usepackage{graphicx}       
\usepackage{tikz-cd}
\usepackage{amsmath,amssymb,amsthm,mathrsfs}
\theoremstyle{plain}
\newtheorem{theorem}{Theorem}

\newtheorem{corollary}[theorem]{Corollary}
\newtheorem{example}[theorem]{Example}
\theoremstyle{definition}

\newtheorem{remark}[theorem]{Remark}

\graphicspath{{media/}}     

\title{Disproving some theorems in Sharma and Chauhan \textit{et al.} (2018, 2021)
\thanks{\textit{This research was conducted at Université d'Artois, La Faculté Jean Perrin in Lens, and was funded by the Science, Technology \& Innovation Funding Authority (STDF); International Cooperation Grants, project number 49294. Ramy Takieldin would like to express his deepest gratitude to Professor André Leroy for his invaluable guidance throughout this project.
}} 
}

\author{
  Ramy Takieldin \\
  Faculty of Engineering, Ain Shams University, Cairo, Egypt\\
  Egypt University of Informatics, New Capital, Cairo, Egypt\\
  \texttt{ramy.farouk@eng.asu.edu.eg} \\
   \AND
   Patrick Solé \\
   I2M (CNRS, University of Aix-Marseille), 13009 Marseilles, France  \\
   \texttt{patrick.sole@telecom-paris.fr} \\
}

\begin{document}
\maketitle
\sloppy

\begin{abstract}
The main objective of this work is to show, through counterexamples, that some of the theorems presented in the papers of Sharma \textit{et al.} (2018)  and Chauhan \textit{et al.} ( 2021) are incorrect. Although they used these theorems to establish a sufficient condition for a multi-twisted (MT) code to be linear complementary dual (LCD), we show that this condition itself remains valid. We further improve this condition by removing the restrictions on the shift constants and relaxing the required coprimality condition. We show that compared to the previous condition, the modified condition is able to identify more LCD MT codes. Furthermore, without the need for a normalized set of generators, we develop a formula to determine the dimension of any $\rho$-generator MT code.
\end{abstract}

\keywords{ Multi-twisted code \and linear complementary dual \and Determinantal divisors \and Algebraic coding}

{\bf MSC:} 94B05, 94B60, 11T71

\section{Introduction}
Multi-twisted (MT) codes over a finite field $\mathbb{F}_q$ constitute a significant and comprehensive class of linear codes. This class contains several well-known subclasses, including cyclic, constacyclic, quasi-cyclic, quasi-twisted, and generalized quasi-cyclic codes. For some integer $\ell\geq 1$, let $0 \neq \lambda_i \in \mathbb{F}_q$ and $m_i \geq 1$ for $1 \leq i \leq \ell$. If $\Lambda = (\lambda_1, \lambda_2, \ldots, \lambda_\ell)$, then a $\Lambda$-MT code $\mathcal{C}$ with block lengths $(m_1, m_2, \ldots, m_\ell)$ is defined in \cite[Definition 3.1]{Sharma2018} as a linear code of length $n = m_1 + m_2 + \cdots + m_\ell$ that remains invariant under the $\Lambda$-MT linear transformation 
\begin{equation*}\begin{split}
T_{\Lambda}: &\left(c_{1,0}, c_{1,1}, \ldots, c_{1, m_1-1} ; c_{2,0}, c_{2,1}, \ldots, c_{2, m_2-1} ; \ldots ; c_{\ell, 0}, c_{\ell, 1}, \ldots, c_{\ell, m_{\ell}-1}\right)\mapsto\\
&\left(\lambda_1 c_{1, m_1-1},c_{1,0}, \ldots, c_{1, m_1-2} ; \lambda_2 c_{2, m_2-1}, c_{2,0}, \ldots, c_{2, m_2-2} ; \ldots ; \lambda_{\ell} c_{\ell, m_{\ell}-1}, c_{\ell, 0}, \ldots, c_{\ell, m_{\ell}-2}\right).
\end{split}\end{equation*}
Throughout this paper, we adopt the same notations as in \cite{Sharma2018, Chauhan2021}. Thus, $\mathcal{C}$ denotes a $\Lambda$-MT code over $\mathbb{F}_q$ with block lengths $\left(m_1, m_2, \ldots, m_\ell\right)$. The Euclidean dual $\mathcal{C}^\perp$ of $\mathcal{C}$ is a $\left(\lambda_1^{-1}, \lambda_2^{-1}, \ldots, \lambda_\ell^{-1}\right)$-MT code with the same block lengths. By using polynomial representation for blocks, $\mathcal{C}$ can be regarded as an $\mathbb{F}_q[x]$-submodule of the $\Lambda$-MT module 
$$V=\bigoplus_{i=1}^\ell\frac{\mathbb{F}_q[x]}{\langle x^{m_i}-\lambda_i\rangle}=\frac{\mathbb{F}_q[x]}{\langle x^{m_1}-\lambda_1\rangle}\oplus \frac{\mathbb{F}_q[x]}{\langle x^{m_2}-\lambda_2\rangle}\oplus \cdots \oplus \frac{\mathbb{F}_q[x]}{\langle x^{m_\ell}-\lambda_\ell\rangle}.$$ 
If $\mathcal{C}$ is generated as an $\mathbb{F}_q[x]$-module by the set $\left\{\mathbf{g}_1, \mathbf{g}_2, \ldots, \mathbf{g}_\rho\right\} \subseteq V$, then $\mathcal{C}$ is referred to as a $\rho$-generator MT code. For $1 \leq i \leq \ell$, let $\pi_i$ denote the projection of $V$ onto $\frac{\mathbb{F}_q[x]}{\langle x^{m_i}-\lambda_i\rangle}$. It is noteworthy that $\pi_i\left(\mathcal{C}\right)\subseteq \frac{\mathbb{F}_q[x]}{\langle x^{m_i}-\lambda_i\rangle}$ is a $\lambda_i$-constacyclic code of length $m_i$. The generator polynomial of $\pi_i\left(\mathcal{C}\right)$ is $g_i\left(x\right)$, and we write $\pi_i\left(\mathcal{C}\right) = \langle g_i\left(x\right) \rangle$, where 
$$g_i\left(x\right)=\mathrm{gcd}\left\{\left(x^{m_i}-\lambda_i\right),\pi_i\left(\mathbf{g}_1\right),\pi_i\left(\mathbf{g}_2\right),\ldots,\pi_i\left(\mathbf{g}_\rho\right)\right\}.$$

In \cite[Section 5]{Sharma2018}, sufficient conditions are provided for MT codes over a finite field $\mathbb{F}_q$ to be LCD, self-orthogonal, and dual-containing. These conditions are proven under the assumption that the block lengths $m_i$ are coprime to $q$. Later, in \cite[Remark 3.1(b)]{Chauhan2021}, it is claimed that the aforementioned conditions remain valid for arbitrary block lengths not necessarily coprime to $q$. The main objective of this paper is to disprove these theorems by presenting counterexamples. In Section \ref{Counter}, we provide several numerical examples that satisfy all the hypotheses of these theorems but contradict their conclusions. These counterexamples demonstrate that the theorems are incorrect, both for block lengths that are coprime to $q$ \cite{Sharma2018} and for block lengths not coprime to $q$ \cite{Chauhan2021}. Specifically, \cite[Theorem 5.5]{Sharma2018} classifies an MT code as LCD, self-orthogonal, dual-containing, or self-dual based on the code dimension, denoted by $\mathrm{dim}_{\mathbb{F}_q}\mathcal{C}$. However, our numerical examples reveal that code dimension alone is insufficient for such classification. Indeed, for different code dimensions, we show the existence of MT codes that are neither LCD, self-orthogonal, nor dual-containing, see Examples \ref{ex-1}, \ref{ex2}, \ref{ex5}, and \ref{ex6}.

In \cite[Corollary 5.1]{Sharma2018}, a sufficient condition for MT codes to be LCD was proven as a direct consequence of \cite[Theorem 5.6]{Sharma2018}. Although, as shown in Section \ref{Counter}, \cite[Theorem 5.6]{Sharma2018} is incorrect, we prove that \cite[Corollary 5.1]{Sharma2018} remains valid. In Theorem \ref{corr2}, we prove a more general form of \cite[Corollary 5.1]{Sharma2018}, offering a significantly less restrictive condition for an MT code to be LCD. Specifically, our generalization improves upon \cite[Corollary 5.1]{Sharma2018} in three key ways. First, we replace the coprimality condition in \cite[Corollary 5.1]{Sharma2018} with a less restricting condition. Second, we eliminate the need for the shift constants $\lambda_i$ to satisfy $\lambda_i \neq \lambda_i^{-1}$ for each $1 \leq i \leq \ell$. Third, we demonstrate that, under the proposed coprimality condition, both the MT code and its dual can be expressed as a direct sum of constacyclic codes, whereas \cite[Theorem 5.2]{Sharma2018} asserts that only the dual can be expressed as a direct sum. In addition, Theorem \ref{corr2} does not impose the restriction that the code block lengths $m_i$ be coprime to $q$. Examples \ref{ex13} and \ref{ex14} illustrate that Theorem \ref{corr2} can identify more LCD codes than \cite[Corollary 5.1]{Sharma2018}.

Finally, in Section \ref{Sdimension}, we aim to strengthen \cite[Theorem 5.1(b)]{Sharma2018} and \cite[Corollary 3.1]{Chauhan2021}. Specifically, \cite[Theorem 5.1(b)]{Sharma2018} proposes a formula to determine the dimension of a $\rho$-generator MT code, but this formula is valid only when $\rho = 1$. The authors show in \cite[Example 5.1]{Sharma2018} that this formula fails when $\rho \geq 2$. Although an alternative formula was suggested in \cite[Corollary 3.1]{Chauhan2021} to determine the dimension of a $\rho$-generator MT code for any $\rho \geq 1$, it has the drawback of requiring a normalized generating set for the MT code. In Theorem \ref{dimension}, we prove a dimension formula that generalizes \cite[Theorem 5.1(b)]{Sharma2018} for any $\rho \geq 1$, without the need to obtain a normalized generating set.

The remaining sections are organized as follows. Section \ref{Counter} presents examples that contradict \cite[Theorems 5.5 and 5.6]{Sharma2018}, both for block lengths that are coprime to $q$ and those that are not. Section \ref{Coprimality} establishes our modified condition for an MT code to be LCD. Section \ref{Sdimension} provides the proof of the dimension formula for any $\rho$-generator MT code.

\section{Counterexamples to \cite[Section 5]{Sharma2018} and \cite[Remark 3.1(b)]{Chauhan2021}}
\label{Counter}
In this section, we present examples that disprove \cite[Theorems 5.5 and 5.6]{Sharma2018}. The first subsection focuses on the case where the block lengths $m_i$ are coprime to $q$. In contrast, the second subsection considers the case where the block lengths $m_i$ are not coprime to $q$. These examples also show that the code dimension alone is insufficient for classifying an MT code as LCD, self-orthogonal, or dual-containing. To simplify these examples, we use two representations for each MT code: first, as a linear code with a generator matrix, and second, as an $\mathbb{F}_q[x]$-module with generators $\left\{\mathbf{g}_1, \mathbf{g}_2, \ldots, \mathbf{g}_\rho\right\}$.

\subsection{$m_i$ coprime to $q$}
The following example contradicts \cite[Theorem 5.5(a)]{Sharma2018}, which states that any $(\lambda_1, \lambda_2, \ldots, \lambda_\ell)$-MT code with block lengths $(m_1, m_2, \ldots, m_\ell)$, where $\lambda_i \neq \lambda_i^{-1}$ for $1 \leq i \leq \ell$, is LCD if either $\mathrm{dim}_{\mathbb{F}_q}\mathcal{C}$ or $\mathrm{dim}_{\mathbb{F}_q}\mathcal{C}^\perp$ is smaller than $\min{m_i}$, the smallest block length.

\begin{example}
\label{ex-1}
Consider the MDS linear code $\mathcal{C}$ over $\mathbb{F}_5$ of length $n=12$, dimension $k=11$, minimum distance $d_{\mathrm{min}}=2$, and generator matrix
\begin{equation*}
G=\begin{bmatrix}
 1&0&0&0&0&0&0&0&0&0&0&3\\
 0&1&0&0&0&0&0&0&0&0&0&4\\
 0&0&1&0&0&0&0&0&0&0&0&2\\
 0&0&0&1&0&0&0&0&0&0&0&4\\
 0&0&0&0&1&0&0&0&0&0&0&2\\
 0&0&0&0&0&1&0&0&0&0&0&1\\
 0&0&0&0&0&0&1&0&0&0&0&3\\
 0&0&0&0&0&0&0&1&0&0&0&4\\
 0&0&0&0&0&0&0&0&1&0&0&2\\
 0&0&0&0&0&0&0&0&0&1&0&1\\
 0&0&0&0&0&0&0&0&0&0&1&3
\end{bmatrix}.
\end{equation*}
In fact, $\mathcal{C}$ is the $\left(2,3\right)$-MT code over $\mathbb{F}_5$ with block lengths $\left(3,9\right)$ and generators 
\begin{equation*}
\begin{split}
\mathbf{g}_1&=\left(1+4x+3x^2 , 4+2x+3x^2+x^3+x^4+x^5+x^6+3x^8\right),\\
\mathbf{g}_2&=\left( 1+4x, 3+4x^3+2x^6+ 2x^7+x^8 \right).
\end{split}
\end{equation*}
That is, $\ell=2$, $m_1=3$, $m_2=9$, $\lambda_1=2$, and $\lambda_2=3$. We have $m_i$ is coprime to $q$, $\lambda_i\ne\lambda_i^{-1}$ for $1\le i\le 2$, and $\mathrm{dim}_{\mathbb{F}_q}\mathcal{C}^\perp=1 < \min\{m_i\}=3$. Then \cite[Theorem 5.5(a)]{Sharma2018} implies that $\mathcal{C}$ is an LCD code. But $\mathcal{C}$ is not LCD because of any of the following reasons:
\begin{enumerate}
\item By \cite[Theorem 2.4]{Liu2018}, $\mathcal{C}$ is not LCD because $GG^T$ is singular.
\item The non-zero codeword $\left(  1,3,4,3,4,2,1,3,4,2,1,3\right)$ spans the one-dimensional code $\mathcal{C}\cap\mathcal{C}^\perp$.
\end{enumerate}
\end{example}

The following example contradicts \cite[Theorem 5.5(b-c)]{Sharma2018}. In particular, \cite[Theorem 5.5(b)]{Sharma2018} states that any $(\lambda_1, \lambda_2, \ldots, \lambda_\ell)$-MT code $\mathcal{C}$ with block lengths $(m_1, m_2, \ldots, m_\ell)$, where $\lambda_i \neq \lambda_i^{-1}$ for $1 \leq i \leq \ell$, is LCD or self-orthogonal if $\mathrm{dim}_{\mathbb{F}_q}\mathcal{C}=\min{m_i}$, the smallest block length. While \cite[Theorem 5.5(c)]{Sharma2018} states that $\mathcal{C}$ is LCD or dual-containing if $\mathrm{dim}_{\mathbb{F}_q}\mathcal{C}^\perp=\min{m_i}$.

\begin{example}
\label{ex2}
Consider the linear code $\mathcal{C}$ over $\mathbb{F}_9$ (= $\mathbb{F}_3(\omega)$ with $\omega^2+2\omega+2=0 $) of length $n=16$, dimension $k=4$, minimum distance $d_{\mathrm{min}}=10$, and generator matrix
\begin{equation*}
G=\begin{bmatrix}
   1 & 0 & 0 & 0&\omega^3&\omega^7&\omega^6 & 0&\omega^7&\omega^5&\omega^7&\omega^5 & 1&\omega^6 & 1&\omega^6\\
   0 & 1 & 0 & 0 & 0&\omega^3&\omega^7&\omega^6 & 2&\omega^7&\omega^5&\omega^7&\omega^5 & 1&\omega^6 & 1\\
   0 & 0 & 1 & 0&\omega^5 & 0&\omega^3&\omega^7&\omega^6 & 2&\omega^7&\omega^5&\omega^7&\omega^5 & 1&\omega^6\\
   0 & 0 & 0 & 1&\omega^6&\omega^5 & 0&\omega^3 & 2&\omega^6 & 2&\omega^7&\omega^5&\omega^7&\omega^5 & 1
\end{bmatrix}.
\end{equation*}
In fact, $\mathcal{C}$ is a $\left(\omega^7, \omega^7, \omega^6\right)$-MT code over $\mathbb{F}_9$ with block lengths $\left(4,4,8\right)$ and generated by
\begin{equation*}
\begin{split}
\mathbf{g}&=\left(1 ,   \omega^3 + \omega^7 x +\omega^6 x^2 ,   \omega^7 + \omega^5 x  + \omega^7 x^2 + \omega^5 x^3 + x^4+ \omega^6 x^5 + x^6 +\omega^6 x^7\right).
\end{split}
\end{equation*}
That is, $\ell=3$, $m_1=m_2=4$, $m_3=8$, $\lambda_1=\lambda_2=\omega^7$, and $\lambda_3=\omega^6$. We have $m_i$ is coprime to $q$, $\lambda_i\ne\lambda_i^{-1}$ for $1\le i\le 3$, and $\mathrm{dim}_{\mathbb{F}_9}\mathcal{C}=4 =\min\{m_i\}$. Then \cite[Theorem 5.5(b)]{Sharma2018} implies that $\mathcal{C}$ is either an LCD or a self-orthogonal code, which is wrong. In fact, \cite[Theorem 2.4]{Liu2018} shows that $\mathcal{C}$ is not LCD because  
\begin{equation*}
GG^T=\begin{bmatrix}
 \omega^5&\omega^7&  0&\omega^3\\
 \omega^7&\omega^7 & 1&\omega^5\\
   0 & 1&\omega^7 & 0\\
 \omega^3&\omega^5 & 0&\omega^1
\end{bmatrix}
\end{equation*}
is singular. But $GG^T$ is not the zero matrix, hence $\mathcal{C}$ is not self-orthogonal. In fact, machine computation shows that the hull is one-dimensional spanned by the vector

$$[  1,   1, \omega^7, \omega^3,   0,   1, \omega^7,   0,   \omega,   0,   0, \omega^7, \omega^2,   0,   0,   1].$$

We also remark that $\mathcal{C}^\perp$ satisfies all conditions of \cite[Theorem 5.5(c)]{Sharma2018}, i.e., $\mathrm{dim}_{\mathbb{F}_9}\left(\mathcal{C}^\perp\right)^\perp=4 =\min\{m_i\}$. But $\mathcal{C}^\perp$ is neither LCD nor dual-containing, because $\mathcal{C}$ is neither LCD nor self-orthogonal. This contradicts \cite[Theorem 5.5(c)]{Sharma2018}.
\end{example}

The following example contradicts \cite[Theorem 5.6]{Sharma2018}, which states that any $(\lambda_1, \lambda_2, \ldots, \lambda_\ell)$-MT code $\mathcal{C}$ with block lengths $(m_1, m_2, \ldots, m_\ell)$, where $\lambda_i \neq \lambda_i^{-1}$ for $1 \leq i \leq \ell$, is LCD if $\pi_i\left(\mathcal{C}\right)\ne\langle 1\rangle$ or $\pi_i\left(\mathcal{C}^\perp\right)\ne\langle 1\rangle$ for each $1\le i\le \ell$.
\begin{example}
\label{ex-9}
Consider the same MT code $\mathcal{C}$ used in Example \ref{ex-1}. We have $\mathcal{C}^\perp$ is an MDS $\left(3,2\right)$-MT code over $\mathbb{F}_5$ with block lengths $\left(3,9\right)$, dimension $n-k=1$, minimum distance $d^\perp_{\mathrm{min}}=12$, and generator matrix
\begin{equation*}
H=\begin{bmatrix}
 1&3&4&3&4&2&1&3&4&2&1&3
\end{bmatrix}.
\end{equation*}
According to \cite[Theorem 5.6]{Sharma2018}, since $\pi_1\left(\mathcal{C}^\perp\right)=\langle 4+2x+x^2 \rangle\ne\langle 1\rangle$ and $\pi_2\left(\mathcal{C}^\perp\right)=\langle 1+3x+4x^2+2x^3+x^4+3x^5+4x^6+2x^7+x^8 \rangle\ne\langle 1\rangle$, $\mathcal{C}$ is LCD. However, $\mathcal{C}$ is not LCD as we showed in Example \ref{ex-1}. We remark that $\pi_1\left(\mathcal{C}^\perp\right)$ and $\pi_2\left(\mathcal{C}^\perp\right)$ are MDS codes.
\end{example}

\subsection{$m_i$ not coprime to $q$}
In \cite[Remark 3.1(b)]{Chauhan2021}, it is claimed that the theorems in \cite[Section 5]{Sharma2018} hold even for arbitrary block lengths, not necessarily coprime to $q$. In this subsection, we demonstrate that this assertion is incorrect. The following example contradicts \cite[Theorem 5.5(a)]{Sharma2018} for block lengths not coprime to $q$.

\begin{example}
\label{ex5}
Consider the linear code $\mathcal{C}$ over $\mathbb{F}_5$ of length $n=10$, dimension $k=2$, minimum distance $d_{\mathrm{min}}=8$, and generator matrix
\begin{equation*}
G=\begin{bmatrix}
 1&0&1&4&2&0&2&3&4&4\\
 0&1&4&2&2&2&3&4&4&0
\end{bmatrix}.
\end{equation*}
In fact, $\mathcal{C}$ is a $\left(3,3\right)$-MT code over $\mathbb{F}_5$ with block lengths $\left(5,5\right)$ generated by
\begin{equation*}
\begin{split}
\mathbf{g}&=\left( 3+ 2x + x^2 +x^3  ,  4+ 2x+ 2x^2 +2x^4 \right).
\end{split}
\end{equation*}
That is, $\ell=2$, $m_1=m_2=q=5$, $\lambda_1=\lambda_2=3$. We have $m_i$ is not coprime to $q$, $\lambda_i\ne\lambda_i^{-1}$ for $1\le i\le 2$, and $\mathrm{dim}_{\mathbb{F}_q}\mathcal{C}=2 < \min\{m_i\}=5$. As been mentioned in \cite[Remark 3.1(b)]{Chauhan2021}, such conditions imply that $\mathcal{C}$ is an LCD code. This is wrong because of any of the following reasons:
\begin{enumerate}
\item By \cite[Theorem 2.4]{Liu2018}, $\mathcal{C}$ is not LCD because $GG^T$ is singular.
\item The non-zero codeword $\left(  0, 1, 4, 2, 2, 2, 3, 4, 4, 0  \right)\in\mathcal{C}\cap\mathcal{C}^\perp$. In fact it constitutes a basis of the hull.
\end{enumerate}
\end{example}

The following example contradicts \cite[Theorem 5.5(b-c)]{Sharma2018} for block lengths not coprime to $q$.

\begin{example}
\label{ex6}
Consider the linear code $\mathcal{C}$ over $\mathbb{F}_4$ of length $n=12$, dimension $k=4$, minimum distance $d_{\mathrm{min}}=6$, and generator matrix
\begin{equation*}
G=\begin{bmatrix}
   1 & 0 & 0 & 0&\omega&\omega & 1&\omega^2&\omega^2&\omega^2&\omega & 1\\
   0 & 1 & 0 & 0&\omega&\omega&\omega & 1&\omega^2&\omega^2&\omega^2&\omega\\
   0 & 0 & 1 & 0&\omega^2&\omega&\omega&\omega & 1&\omega^2&\omega^2&\omega^2\\
   0 & 0 & 0 & 1 & 1&\omega^2&\omega&\omega&\omega & 1&\omega^2&\omega^2
\end{bmatrix}.
\end{equation*}
In fact, $\mathcal{C}$ is an $\left(\omega^2,\omega\right)$-MT code over $\mathbb{F}_4$ with block lengths $\left(4,8\right)$ and generated by
\begin{equation*}
\begin{split}
\mathbf{g}&=\left(1 ,  \omega+ \omega x + x^2 + \omega^2 x^3 + \omega^2 x^4 + \omega^2 x^5 + \omega x^6 +x^7\right).
\end{split}
\end{equation*}
That is, $\ell=2$, $m_1=4$, $m_2=8$, $\lambda_1=\omega^2$, and $\lambda_2=\omega$. We have $m_i$ is not coprime to $q$, $\lambda_i\ne\lambda_i^{-1}$ for $1\le i\le 2$, and $\mathrm{dim}_{\mathbb{F}_q}\mathcal{C}=4 =\min\{m_i\}$. As been mentioned in \cite[Remark 3.1(b)]{Chauhan2021}, having $\mathcal{C}$ satisfying all conditions of \cite[Theorem 5.5(b)]{Sharma2018} which implies is either an LCD or a self-orthogonal code, which is wrong. In fact, \cite[Theorem 2.4]{Liu2018} shows that $\mathcal{C}$ is not LCD because $GG^T$ is singular. But $GG^T$ is not the zero matrix, hence $\mathcal{C}$ is not self-orthogonal. Machine calculations show that the hull is one-dimensional spanned by
$$[  1,   0,   0, \omega^2,   1,   0,   0,   \omega,   \omega,   0,   0, \omega^2].$$

We remark that $\mathcal{C}^\perp$ satisfies all conditions of \cite[Theorem 5.5(c)]{Sharma2018}, then \cite[Remark 3.1(b)]{Chauhan2021} implies that $\mathcal{C}^\perp$ is LCD or dual-containing. This is not true since $\mathcal{C}$ is neither LCD nor self-orthogonal.
\end{example}

The following example contradicts \cite[Theorem 5.6]{Sharma2018} for block lengths not coprime to $q$.
\begin{example}
\label{ex10}
Consider the $\left(3,3\right)$-MT code $\mathcal{C}$ over $\mathbb{F}_5$ with block lengths $\left(5,5\right)$ presented in Example \ref{ex5}. We noted that $\pi_1\left(\mathcal{C}\right)=\pi_2\left(\mathcal{C}\right)=\langle (x+2)^3 \rangle$ are $[5,2,4]$ MDS $3$-constacyclic codes. According to \cite[Theorem 5.6]{Sharma2018} and \cite[Remark 3.1(b)]{Chauhan2021}, since $\pi_1\left(\mathcal{C}\right)\ne\langle 1\rangle$ and $\pi_2\left(\mathcal{C}\right) \ne\langle 1\rangle$, $\mathcal{C}$ is LCD. However, $\mathcal{C}$ is not LCD as we showed in Example \ref{ex5}. 
\end{example}

\section{Sufficient condition for LCD}
\label{Coprimality}
In the previous section, we showed by counterexamples that \cite[Theorem 5.5]{Sharma2018} and \cite[Theorem 5.6]{Sharma2018} are wrong. In addition, these examples emphasize that imposing a condition on the dimension of the MT code or its dual to guarantee being LCD is impossible. For instance, the MT code given in Example \ref{ex-1} has dimension greater than $\min\{m_i\}$, while its dual has dimension smaller than $\min\{m_i\}$, and that in Example \ref{ex2} has dimension equal to $\min\{m_i\}$, and all of which are not LCD. In this section, we aim to present a sufficient condition for an MT code to be LCD. In \cite[Corollary 5.1]{Sharma2018}, such a sufficient condition for an MT to be LCD is given. However, \cite[Corollary 5.1]{Sharma2018} has been proven as a consequence of \cite[Theorem 5.6]{Sharma2018}. Although Examples \ref{ex-9} and \ref{ex10} showed that \cite[Theorem 5.6]{Sharma2018} is wrong, we will prove that \cite[Corollary 5.1]{Sharma2018} is correct in a more general context. In fact, we will prove a sufficient condition for an MT code to be LCD which is much less restricting than that of \cite[Corollary 5.1]{Sharma2018}. To this end, we start by proving a more general result than \cite[Theorem 5.2]{Sharma2018} with a less restricting coprimality condition. Specifically, \cite[Theorem 5.2]{Sharma2018} shows that if $x^{m_1}-\lambda_1, x^{m_2}-\lambda_2, \cdots, x^{m_{\ell}}-\lambda_{\ell}$ are pairwise coprime, then the dual of the a $(\lambda_1,\ldots,\lambda_\ell)$-MT code of block lengths $(m_1,\ldots,m_\ell)$ is the direct sum of $\ell$ constacyclic codes. In the following result, we show that both $\mathcal{C}$ and its dual are the direct sum of $\ell$ constacyclic codes under a less restricting coprimality condition.

\begin{theorem}
\label{directprod}
Let $\Lambda=\left(\lambda_1,\lambda_2,\ldots,\lambda_\ell\right)$, where $\lambda_i$ is a non-zero element of $\mathbb{F}_q$ for $1\le i\le \ell$. Let $\mathcal{C}$ (or $\mathcal{C}^\perp$) be a $\Lambda$-MT code of block lengths $\left(m_1,m_2,\ldots,m_\ell\right)$ generated by $\left\{\mathbf{g}_1, \mathbf{g}_2, \ldots, \mathbf{g}_\rho\right\}$. For $1\le i\le \ell$, let $g_i\left(x\right)=\mathrm{gcd}\left\{\left(x^{m_i}-\lambda_i\right),\pi_i\left(\mathbf{g}_1\right),\pi_i\left(\mathbf{g}_2\right),\ldots,\pi_i\left(\mathbf{g}_\rho\right)\right\}$. If 
$$\frac{x^{m_1}-\lambda_1}{g_1\left(x\right)}, \frac{x^{m_2}-\lambda_2}{g_2\left(x\right)}, \ldots, \frac{x^{m_\ell}-\lambda_\ell}{g_\ell\left(x\right)}$$ are pairwise coprime polynomials in $\mathbb{F}_q[x]$, then $\mathcal{C}$ and $\mathcal{C}^\perp$ are direct sums of $\ell$ constacyclic codes, in particular
\begin{equation*}
\begin{split}
\mathcal{C}&=\bigoplus_{i=1}^\ell \pi_i\left(\mathcal{C}\right)\\
\mathcal{C}^\perp&=\bigoplus_{i=1}^\ell \pi_i\left(\mathcal{C}\right)^\perp=\bigoplus_{i=1}^\ell \pi_i\left(\mathcal{C}^\perp\right).
\end{split}
\end{equation*}
\end{theorem}
\begin{proof}
Clearly $\mathcal{C}$ is LCD if and only if $\mathcal{C}^\perp$ is LCD, and hence it suffices to prove the result for $\mathcal{C}$. Let $\mathbf{c}=\left(c_1\left(x\right),c_2\left(x\right),\ldots,c_\ell\left(x\right)\right)\in\mathcal{C}$. Then
$$\mathbf{c}=\sum_{i=1}^\ell \left(0,\ldots,0,\pi_i\left(\mathbf{c}\right),0,\ldots,0\right)\in \bigoplus_{i=1}^\ell \pi_i\left(\mathcal{C}\right).$$
Thus, $\mathcal{C}\subseteq \bigoplus_{i=1}^\ell \pi_i\left(\mathcal{C}\right)$. Conversely, let $\left(a_1\left(x\right),a_2\left(x\right),\ldots,a_\ell\left(x\right)\right)\in \bigoplus_{i=1}^\ell \pi_i\left(\mathcal{C}\right)$. For a fixed $i\in\{1,\ldots, \ell\}$, we have $a_i\left(x\right)\in \pi_i\left(\mathcal{C}\right)$, hence $a_i\left(x\right)= g_i\left(x\right) b_i\left(x\right)$ for some $b_i\left(x\right)\in\mathbb{F}_q[x]$. Moreover, there exists a codeword $\mathbf{c}^{(i)}\in\mathcal{C}$ such that $\pi_i\left(\mathbf{c}^{(i)}\right)=a_i\left(x\right)$. The Chinese remainder theorem ensures the existence a polynomial $f_i\left(x\right)\in\mathbb{F}_q[x]$ such that 
$$ f_i\left(x\right)\equiv 1 \pmod{\frac{x^{m_i}-\lambda_i}{g_i\left(x\right)}} \quad \text{and} \quad f_i\left(x\right)\equiv 0 \pmod{\frac{x^{m_j}-\lambda_j}{g_j\left(x\right)}} \ \forall j\ne i.$$
Since $g_i\left(x\right)| \pi_i\left(\mathbf{c}^{(i)}\right)$ and $g_j\left(x\right)| \pi_j\left(\mathbf{c}^{(i)}\right)$, we have 
\begin{equation*}\begin{split}
f_i\left(x\right)\pi_i\left(\mathbf{c}^{(i)}\right)&
\equiv \pi_i\left(\mathbf{c}^{(i)}\right) \pmod{x^{m_i}-\lambda_i} \text{ and }\\
f_i\left(x\right)\pi_j\left(\mathbf{c}^{(i)}\right)&
\equiv 0 \pmod{x^{m_j}-\lambda_j} \quad \forall j\ne i.
\end{split}\end{equation*}
Consequently,
\begin{equation}
\label{Indirectprod}
\begin{split}
\left(a_1\left(x\right),a_2\left(x\right),\ldots,a_\ell\left(x\right)\right)&=\sum_{i=1}^\ell \left(0,\ldots,0,a_i\left(x\right),0,\ldots,0\right)\\
&=\sum_{i=1}^\ell \left(0,\ldots,0,\pi_i\left(\mathbf{c}^{(i)}\right),0,\ldots,0\right)\\
&=\sum_{i=1}^\ell f_i\left(x\right) \mathbf{c}^{(i)} \in\mathcal{C},
\end{split}
\end{equation}
this is because $\mathcal{C}$ is an $\mathbb{F}_q[x]$-module, $f_i\left(x\right)\in\mathbb{F}_q[x]$, and $\mathbf{c}^{(i)} \in\mathcal{C}$ for each $1\le i\le \ell$. That is, $ \bigoplus_{i=1}^\ell \pi_i\left(\mathcal{C}\right)\subseteq\mathcal{C}$. Therefore,
\begin{equation*}
\mathcal{C}=\bigoplus_{i=1}^\ell \pi_i\left(\mathcal{C}\right)\quad \text{and}\quad\mathcal{C}^\perp=\bigoplus_{i=1}^\ell \pi_i\left(\mathcal{C}\right)^\perp.
\end{equation*}
Moreover, $$\pi_i\left(\mathcal{C}^\perp\right)=\pi_i\left(\bigoplus_{i=1}^\ell \pi_i\left(\mathcal{C}\right)^\perp\right)=\pi_i\left(\mathcal{C}\right)^\perp.$$
\end{proof}

Theorem \ref{directprod} not only weaken the coprimality requirement of \cite[Theorem 5.2]{Sharma2018}, but it shows that if such condition is satisfied for $\mathcal{C}$ or its dual then both are the direct sum of $\ell$ constacyclic codes. Now, a direct consequence of Theorem \ref{directprod} yields a sufficient condition for an MT code to be LCD which shows the correctness of \cite[Corollary 5.1]{Sharma2018} and moreover generalizes it.

\begin{corollary}
\label{corr}
Let $\Lambda=\left(\lambda_1,\lambda_2,\ldots,\lambda_\ell\right)$, where $0\ne\lambda_i\ne\lambda_i^{-1}$ for $1\le i\le \ell$. Let $\mathcal{C}$ (or $\mathcal{C}^\perp$) be a $\Lambda$-MT code of block lengths $\left(m_1,m_2,\ldots,m_\ell\right)$ generated by $\left\{\mathbf{g}_1, \mathbf{g}_2, \ldots, \mathbf{g}_\rho\right\}$. For $1\le i\le \ell$, let $g_i\left(x\right)=\mathrm{gcd}\left\{\left(x^{m_i}-\lambda_i\right),\pi_i\left(\mathbf{g}_1\right),\pi_i\left(\mathbf{g}_2\right), \ldots,\pi_i\left(\mathbf{g}_\rho\right)\right\}$. Assume 
$$\frac{x^{m_1}-\lambda_1}{g_1\left(x\right)}, \frac{x^{m_2}-\lambda_2}{g_2\left(x\right)}, \ldots, \frac{x^{m_\ell}-\lambda_\ell}{g_\ell\left(x\right)}$$ are pairwise coprime polynomials in $\mathbb{F}_q[x]$, then $\mathcal{C}$ is LCD.
\end{corollary}
\begin{proof}
From Theorem \ref{directprod}, $\pi_i\left(\mathcal{C}^\perp\right)=\pi_i\left(\mathcal{C}\right)^\perp$. Let $\left(c_1\left(x\right),c_2\left(x\right),\ldots,c_\ell\left(x\right)\right)\in\mathcal{C}\cap\mathcal{C}^\perp$. Then, for each $1\le i\le \ell$, we have $c_i\left(x\right)\in\pi_i\left(\mathcal{C}\right)\cap\pi_i\left(\mathcal{C}^\perp\right)=\pi_i\left(\mathcal{C}\right)\cap\pi_i\left(\mathcal{C}\right)^\perp$.  But $\pi_i\left(\mathcal{C}\right)$ is a $\lambda_i$-constacyclic of length $m_i$ with $\lambda_i\ne\lambda_i^{-1}$, hence $\pi_i\left(\mathcal{C}\right)$ is LCD by \cite[Corollary 3.3]{Liu2018}. That is, $c_i\left(x\right)=0$ for each $1\le i\le \ell$. 
\end{proof}

It is clear now that \cite[Corollary 5.1]{Sharma2018} is just a special case of Corollary \ref{corr}, since $\frac{x^{m_i}-\lambda_i}{g_i\left(x\right)}$ are pairwise coprime if $x^{m_i}-\lambda_i$ are pairwise coprime, but the converse is not true, see Example \ref{ex13}. Thus, Corollary \ref{corr} provides a sufficient condition for an MT code to be LCD with a less restricting coprimality condition than \cite[Corollary 5.1]{Sharma2018}. We aim at adjusting this condition to treat the case if $0\ne\lambda_i = \lambda_i^{-1}$ for some $1\le i\le \ell$. In addition, we don't restrict $m_i$ to be coprime to $q$. Thus, the following result generalizes Corollary \ref{corr} by not requiring that $0\ne\lambda_i\ne\lambda_i^{-1}$ while keeping the less restricting coprimality condition. We remark that a polynomial $g\left(x\right)$ is self-reciprocal if it is associate to its reciprocal polynomial $g^\star\left(x\right)=x^{\deg{g}}g\left(\frac{1}{x}\right)$, i.e., there exists a nonzero constant $\alpha\in\mathbb{F}_q$ such that  $g\left(x\right)=\alpha x^{\deg{g}}g\left(\frac{1}{x}\right)$.
\begin{theorem}
\label{corr2}
Let $\Lambda=\left(\lambda_1,\lambda_2,\ldots,\lambda_\ell\right)$, where $\lambda_i$ is a non-zero element of $\mathbb{F}_q$ for $1\le i\le \ell$. Let $\mathcal{C}$ (or $\mathcal{C}^\perp$) be a $\Lambda$-MT code of block lengths $\left(m_1,m_2,\ldots,m_\ell\right)$ generated by $\left\{\mathbf{g}_1, \mathbf{g}_2, \ldots, \mathbf{g}_\rho\right\}$. For $1\le i\le \ell$, let $g_i\left(x\right)=\mathrm{gcd}\left\{\left(x^{m_i}-\lambda_i\right),\pi_i\left(\mathbf{g}_1\right),\pi_i\left(\mathbf{g}_2\right), \ldots,\pi_i\left(\mathbf{g}_\rho\right)\right\}$. If 
$$\frac{x^{m_1}-\lambda_1}{g_1\left(x\right)}, \frac{x^{m_2}-\lambda_2}{g_2\left(x\right)}, \ldots, \frac{x^{m_\ell}-\lambda_\ell}{g_\ell\left(x\right)}$$ are pairwise coprime polynomials in $\mathbb{F}_q[x]$, then $\mathcal{C}$ is LCD if and only if $g_i\left(x\right)$ is self-reciprocal coprime to $\frac{x^{m_i}-\lambda_i}{g_i\left(x\right)}$ for every $i$ with $\lambda_i^2=1$. 
\end{theorem}
\begin{proof}
Assume $\left(c_1\left(x\right),c_2\left(x\right),\ldots,c_\ell\left(x\right)\right)\in\mathcal{C}\cap\mathcal{C}^\perp$. For any $1\le i\le \ell$ such that $\lambda_i^2\ne 1$, we have $c_i\left(x\right)=0$ similar to the proof of Corollary \ref{corr}. On the other hand, for any $1\le i\le \ell$ such that $\lambda_i^2= 1$, we have $\pi_i\left(\mathcal{C}\right)$ and $\pi_i\left(\mathcal{C}\right)^\perp$ are $\lambda_i$-constacyclic codes of length $m_i$. In this case, $\pi_i\left(\mathcal{C}\right)=\langle g_i\left(x\right) \rangle$ is LCD if and only if $g_i\left(x\right)$ is self-reciprocal coprime to $\frac{x^{m_i}-\lambda_i}{g_i\left(x\right)}$. To see this, we have $\pi_i\left(\mathcal{C}\right) \cap \pi_i\left(\mathcal{C}\right)^\perp$ is $\lambda_i$-constacyclic of length $m_i$ generated by $\mathrm{lcm}\left\{g_i\left(x\right), \frac{x^{m_i}-\lambda_i}{g_i^\star\left(x\right)}\right\}$. But $\mathrm{lcm}\left\{g_i\left(x\right), \frac{x^{m_i}-\lambda_i}{g_i^\star\left(x\right)}\right\}=x^{m_i}-\lambda_i$ if and only if $g_i\left(x\right)$ is self-reciprocal coprime to $\frac{x^{m_i}-\lambda_i}{g_i\left(x\right)}$.
\end{proof}

Examples \ref{ex13}, \ref{ex14} show that the conditions in Theorem \ref{corr2} can detect LCD MT codes more than what \cite[Corollary 5.1]{Sharma2018} can detect. Specifically, Example \ref{ex13} presents an MT code for which $\frac{x^{m_i}-\lambda_i}{g_i\left(x\right)}$ are pairwise coprime but $x^{m_i}-\lambda_i$ are not, and hence it is LCD since $\lambda_i^2\ne 1$ for all $\le i\le \ell$. On the other side, Example \ref{ex14} presents an MT code for which $x^{m_i}-\lambda_i$ are pairwise coprime but $0\ne\lambda_i=\lambda_i^{-1}$. Although, \cite[Corollary 5.1]{Sharma2018} can not be applied, Theorem \ref{corr2} is applicable. 

\begin{example}
\label{ex13}
Consider the linear code $\mathcal{C}$ over $\mathbb{F}_4$ of length $n=10$, dimension $k=5$, minimum distance $d_{\mathrm{min}}=3$, and generator matrix
\begin{equation*}
G=\begin{bmatrix}
   1 & 0 & 0 & 1 & 1 & 0 & 0 & 0 & 0 & 0\\
   0 & 1 & 0&\omega&\omega^2 & 0 & 0 & 0 & 0 & 0\\
   0 & 0 & 1 & 1&\omega^2 & 0 & 0 & 0 & 0 & 0\\
   0 & 0 & 0 & 0 & 0 & 1 & 0&\omega^2&\omega^2 & 1\\
   0 & 0 & 0 & 0 & 0 & 0 & 1 & 1&\omega & 1
\end{bmatrix}.
\end{equation*}
In fact, $\mathcal{C}$ is an $\left(\omega, \omega \right)$-MT code over $\mathbb{F}_4$ with block lengths $\left(5,5\right)$ and generated by
\begin{equation*}
\begin{split}
\mathbf{g}&=\left( 1+x+ \omega^2 x^2, 1 + x+ \omega x^2+ x^3 \right).
\end{split}
\end{equation*}
That is, $\ell=2$, $m_1=m_2=5$, $\lambda_1=\lambda_2=\omega$. In fact, \cite[Corollary 5.1]{Sharma2018} is not applicable to $\mathcal{C}$ because $x^{m_1}-\lambda_1$ is not coprime to $x^{m_2}-\lambda_2$ in $\mathbb{F}_4[x]$. However, $\mathcal{C}$ is LCD by Theorem \ref{corr2}. To see this, we have $g_1\left(x\right)=\mathrm{gcd}\left\{1+x+ \omega^2 x^2, x^5-\omega\right\}=\omega+\omega x+x^2$ and $g_2\left(x\right)=\mathrm{gcd}\left\{1 + x+ \omega x^2+ x^3, x^5-\omega\right\}=\left(x+\omega^2\right)\left( x^2+x+\omega\right)$. Since 
$$\frac{x^{m_1}-\lambda_1}{g_1\left(x\right)}=  \left(x+\omega^2\right) \left(x^2+x+\omega \right)\quad \text{and}\quad \frac{x^{m_2}-\lambda_2}{g_2\left(x\right)}=x^2+\omega x+\omega$$
are coprime, Theorem \ref{corr2} shows that $\mathcal{C}$ is LCD because $\lambda_i^2\ne 1$ for $i=1,2$. This example shows the superiority of Theorem \ref{corr2} over \cite[Corollary 5.1]{Sharma2018}, since \cite[Corollary 5.1]{Sharma2018} is not able to conclude that $\mathcal{C}$ is LCD in this case. As a linear code generated by $G$, one can show that $\mathcal{C}$ is LCD using \cite[Theorem 2.4]{Liu2018}, this is because $GG^T$ is nonsingular. Moreover, Theorem \ref{directprod} ensures that $\mathcal{C}= \pi_1\left(\mathcal{C}\right) \oplus  \pi_2\left(\mathcal{C}\right)$, where $\pi_1\left(\mathcal{C}\right)$ and $\pi_2\left(\mathcal{C}\right)$ are the $\omega$-constacyclic MDS code of length $5$ over $\mathbb{F}_4$ generated by $g_1\left(x\right)$ and $g_2\left(x\right)$, respectively. 
\end{example}

\begin{example}
\label{ex14}
Consider the linear code $\mathcal{C}$ over $\mathbb{F}_3$ of length $n=12$, dimension $k=10$, minimum distance $d_{\mathrm{min}}=2$, and generator matrix
\begin{equation*}
G=\begin{bmatrix}
 1&0&0&0&2&0&0&0&0&0&0&0\\
 0&1&0&0&1&0&0&0&0&0&0&0\\
 0&0&1&0&2&0&0&0&0&0&0&0\\
 0&0&0&1&1&0&0&0&0&0&0&0\\
 0&0&0&0&0&1&0&0&0&0&0&2\\
 0&0&0&0&0&0&1&0&0&0&0&2\\
 0&0&0&0&0&0&0&1&0&0&0&2\\
 0&0&0&0&0&0&0&0&1&0&0&2\\
 0&0&0&0&0&0&0&0&0&1&0&2\\
 0&0&0&0&0&0&0&0&0&0&1&2
\end{bmatrix}.
\end{equation*}
In fact, $\mathcal{C}$ is an $\left(2,1 \right)$-MT code over $\mathbb{F}_3$ with block lengths $\left(5,7\right)$ and generated by
\begin{equation*}
\begin{split}
\mathbf{g}&=\left( 1+x+x^2+2x^3+x^4, 1+2x^2+2x^3+x^4+2x^5+x^6  \right).
\end{split}
\end{equation*}
That is, $\ell=2$, $m_1=5$, $m_2=7$, $\lambda_1=2$, and $\lambda_2=1$. In fact, \cite[Corollary 5.1]{Sharma2018} is not applicable to $\mathcal{C}$ because $\lambda_i =\lambda_i^{-1}$. However, $\mathcal{C}$ is LCD by Theorem \ref{corr2}. To see this, we have 
\begin{equation*}\begin{split}
g_1\left(x\right)&=\mathrm{gcd}\left\{1+x+x^2+2x^3+x^4, x^5-2\right\}= x+1,\\
g_2\left(x\right)&=\mathrm{gcd}\left\{1+2x^2+2x^3+x^4+2x^5+x^6, x^7-1\right\}= x-1.
\end{split}\end{equation*}
Since 
$$\frac{x^{m_1}-\lambda_1}{g_1\left(x\right)}= x^4-x^3+x^2-x+1  \quad \text{and}\quad \frac{x^{m_2}-\lambda_2}{g_2\left(x\right)}= x^6+x^5+x^4+x^3+x^2+x+1$$
are coprime, Theorem \ref{corr2} states that $\mathcal{C}$ is LCD if and only if 
\begin{enumerate}
\item $g_1\left(x\right)$ is self-reciprocal coprime to $\frac{x^{m_1}-\lambda_1}{g_1\left(x\right)}$ since $\lambda_1^2=1$, and
\item $g_2\left(x\right)$ is self-reciprocal coprime to $\frac{x^{m_2}-\lambda_2}{g_2\left(x\right)}$ since $\lambda_2^2=1$.
\end{enumerate}
Both conditions are satisfied, hence $\mathcal{C}$ is LCD. This example shows the superiority of Theorem \ref{corr2} over \cite[Corollary 5.1]{Sharma2018}, since \cite[Corollary 5.1]{Sharma2018} is not able to conclude that $\mathcal{C}$ is LCD in this case because $\lambda_i^2=1$. As a linear code generated by $G$, one can show that $\mathcal{C}$ is LCD using \cite[Theorem 2.4]{Liu2018}, this is because $GG^T$ is nonsingular. Moreover, Theorem \ref{directprod} ensures that $\mathcal{C}= \pi_1\left(\mathcal{C}\right) \oplus  \pi_2\left(\mathcal{C}\right)$, where $\pi_1\left(\mathcal{C}\right)$ is the $2$-constacyclic MDS code of length $5$ over $\mathbb{F}_3$ generated by $g_1\left(x\right)$ and $\pi_2\left(\mathcal{C}\right)$ is the cyclic MDS code of length $7$ over $\mathbb{F}_3$ generated by $g_2\left(x\right)$.
\end{example}

\begin{remark}
Although Theorem \ref{corr2} generalizes \cite[Corollary 5.1]{Sharma2018}, we remark that the conditions there are sufficient but not necessary. For instance, 
\begin{enumerate}
\item Consider the code $\mathcal{C}$ and its dual $\mathcal{C}^\perp$ given in Examples \ref{ex-1} and \ref{ex-9}, respectively. Neither $\mathcal{C}$ nor $\mathcal{C}^\perp$ satisfies the coprimality condition of Theorem \ref{corr2}. For instance, $\mathcal{C}^\perp$ is $\left(3,2\right)$-MT over $\mathbb{F}_5$ with block lengths $\left(3,9\right)$ generated by $\mathbf{g}= \left(1+3x+4x^2,3+4x+2x^2+x^3+3x^4+4x^5+2x^6+x^7+3x^8\right)$. Example \ref{ex-9} shows that $ g_1\left(x\right)= 4+2x+x^2 $ and $g_2\left(x\right)= 1+3x+4x^2+2x^3+x^4+3x^5+4x^6+2x^7+x^8 $. Then the coprimality condition of Theorem \ref{corr2} is not satisfied because
$$\frac{x^{3}-3}{g_1\left(x\right)}= \frac{x^{9}-2}{g_2\left(x\right)}= x+3.$$
However, Example \ref{ex-1} shows that $\mathcal{C}$ is not LCD. 
\item Consider the code $\mathcal{C}$ and its dual $\mathcal{C}^\perp$ given in Examples \ref{ex5} and \ref{ex10}, respectively. Neither $\mathcal{C}$ nor $\mathcal{C}^\perp$ satisfies the coprimality condition of Theorem \ref{corr2}. For instance, $\mathcal{C}$ is $\left(3,3\right)$-MT over $\mathbb{F}_5$ with block lengths $\left(5,5\right)$ generated by $\mathbf{g}= \left( 3+ 2x + x^2 +x^3  ,  4+ 2x+ 2x^2 +2x^4 \right)$. The coprimality condition of Theorem \ref{corr2} is not satisfied because $\pi_1\left(\mathcal{C}\right)=\pi_2\left(\mathcal{C}\right)=\langle (x+2)^3 \rangle=\langle g_i\left(x\right)\rangle$. Thus 
$$\frac{x^{m_1}-\lambda_1}{g_1\left(x\right)}= \frac{x^{m_2}-\lambda_2}{g_2\left(x\right)}=(x+2)^2.$$
However, Example \ref{ex5} shows that $\mathcal{C}$ is not LCD. 

\item Consider the $\left(\omega, \omega \right)$-MT code $\mathcal{C}$ over $\mathbb{F}_4$ with block lengths $\left(5,5\right)$, dimension $k=5$, minimum distance $d_{\mathrm{min}}=5$, and generated by
\begin{equation*}
\begin{split}
\mathbf{g}&=\left( x+x^2+x^3+\omega^2 x^4,   1  + \omega x^2 + \omega x^3+ \omega^2 x^4 \right).
\end{split}
\end{equation*}
One can observe that $\pi_1\left(\mathcal{C}\right)=\pi_2\left(\mathcal{C}\right)=\pi_1\left(\mathcal{C}^\perp\right)=\pi_2\left(\mathcal{C}^\perp\right)=\langle 1\rangle$. Although the coprimality condition of Theorem \ref{corr2} is not satisfied for $\mathcal{C}$ and $\mathcal{C}^\perp$, $\mathcal{C}$ is LCD. To see this, we have $GG^T$ is nonsingular, and hence $\mathcal{C}$ is LCD by \cite[Theorem 2.4]{Liu2018}. 
\end{enumerate}
\end{remark}

\section{MT code dimension for $\rho\ge 1$}
\label{Sdimension}
In this section, we strengthen \cite[Theorem 5.1(b)]{Sharma2018} by proposing a generalized dimension formula that extends its applicability to $\rho$-generator MT code for any $\rho\ge 1$. Unlike \cite[Corollary 3.1]{Chauhan2021}, our result does not require a normalized generating set for the MT code.

\begin{theorem}
\label{dimension}
Let $\Lambda=\left(\lambda_1,\lambda_2,\ldots,\lambda_\ell\right)$, where $\lambda_i$ is a non-zero element of $\mathbb{F}_q$ for $1\le i\le \ell$. Let $\mathcal{C}$ be a $\Lambda$-MT code of length $n$. If $\mathcal{C}$ is generated by $\left\{\mathbf{g}_1, \mathbf{g}_2, \ldots, \mathbf{g}_\rho\right\}$, then the dimension of $\mathcal{C}$ as an $\mathbb{F}_q$-vector space is $$\mathrm{dim}_{\mathbb{F}_q}\mathcal{C}=n-\deg\left(\mathfrak{d}_\ell\left(\mathbf{G}\right)\right),$$
where $\mathfrak{d}_\ell\left(\mathbf{G}\right)$ is the $\ell^\mathrm{th}$ determinantal divisor over $\mathbb{F}_q[x]$ of the matrix 
\begin{equation*}
\mathbf{G}=\begin{bmatrix}
\mathbf{g}_1\\
\vdots\\
\mathbf{g}_\rho\\
\mathrm{diag}\left[x^{m_i}-\lambda_i\right]
\end{bmatrix}.
\end{equation*}
In other words, $\mathfrak{d}_\ell$ is the greatest common divisor of all $\ell\times\ell$ minors of $\mathbf{G}$.
\end{theorem}
\begin{proof}
From \cite[Theorem 3.1]{Chauhan2021}, a normalized generating set $\left\{\mathbf{b}_1,\ldots,\mathbf{b}_\ell\right\}$ for $\mathcal{C}$ exists in the lower triangular form. That is, $\mathbf{b}_i=\left(b_{i,1},b_{i,2},\ldots,b_{i,i},0,\ldots,0\right)$ with $b_{i,j}\in\mathbb{F}_q[x]$ and $b_{i,i}\ne 0$. Define the $\ell\times\ell$ matrix 
\begin{equation*}
\mathbf{B}=\begin{bmatrix}
\mathbf{b}_1\\
\vdots\\
\mathbf{b}_\ell
\end{bmatrix}=\left[b_{i,j}\right].
\end{equation*}
Such normalized generating set can be obtained by applying elementary row operations over $\mathbb{F}_q[x]$ to $\mathbf{G}$. In other words, there is an invertible $(\ell+\rho)\times(\ell+\rho)$ matrix $\mathbf{W}$ such that
$$\begin{bmatrix}
\mathbf{B}\\
\mathbf{0}\\
\vdots\\
\mathbf{0}
\end{bmatrix}=\mathbf{W}\mathbf{G}.$$ 
Since determinantal divisors are invariant (up to multiplication by a constant) under multiplication with invertible matrices, there is a non-zero constant $\alpha\in\mathbb{F}_q$ such that
$$\alpha \ \mathfrak{d}_\ell\left(\mathbf{G}\right)=\mathfrak{d}_\ell\left(\begin{bmatrix}
\mathbf{B}\\
\mathbf{0}\\
\vdots\\
\mathbf{0}
\end{bmatrix}\right)=\mathfrak{d}_\ell\left(\mathbf{B}\right)=\mathrm{det}\left(\mathbf{B}\right)=b_{1,1} b_{2,2}\cdots b_{\ell,\ell}.$$
Now the result follows from \cite[Corollary 3.1]{Chauhan2021} because
$$\mathrm{dim}_{\mathbb{F}_q}\mathcal{C}=n-\sum_{i=1}^{\ell}\deg\left(b_{i,i}\right)=n-\deg\left(b_{1,1} b_{2,2}\cdots b_{\ell,\ell}\right)=n-\deg\left(\mathfrak{d}_\ell\left(\mathbf{G}\right)\right).$$
\end{proof}

Theorem \ref{dimension} has the advantage over \cite[Corollary 3.1]{Chauhan2021} that it does not require the construction of a normalized generating set, instead, any generating set for $\mathcal{C}$ can be used. In addition, it generalizes \cite[Theorem 5.1(b)]{Sharma2018} as it is applicable for any $\rho$-generator MT code with $\rho \ge 1$. Moreover, it applies to arbitrary block lengths, as demonstrated in the following examples.
\begin{example}
\label{ex-11}
In Example \ref{ex-1}, we consider the $\left(2,3\right)$-MT MDS code over $\mathbb{F}_5$ with block lengths $\left(3,9\right)$ and generators
\begin{equation*}  
\begin{split}
\mathbf{g}_1&=\left( 1+4x+3x^2 , 4+2x+3x^2+x^3+x^4+x^5+x^6+3x^8 \right),\\
\mathbf{g}_2&=\left( 1+4x, 3+4x^3+2x^6+ 2x^7+x^8  \right).
\end{split}
\end{equation*}
Although $\mathcal{C}$ is a $2$-generator MT code and these generators are not forming a normalized set, we may use Theorem \ref{dimension} to determine the code dimension. We have $\mathrm{dim}_{\mathbb{F}_5}\mathcal{C}=12-\deg\left(\mathfrak{d}_2\left(\mathbf{G}\right)\right),$ where $\mathfrak{d}_2\left(\mathbf{G}\right)$ is the greatest common divisor of all $2\times 2$ minors of  
\begin{equation*}
\mathbf{G}=\begin{bmatrix}
\mathbf{g}_1\\
\mathbf{g}_2\\
\mathrm{diag}\left[x^{m_i}-\lambda_i\right]
\end{bmatrix}=\begin{bmatrix}
1+4x+3x^2 & 4+2x+3x^2+x^3+x^4+x^5+x^6+3x^8 \\
1+4x & 3+4x^3+2x^6+ 2x^7+x^8 \\
x^{3}-2 & 0\\
0 & x^{9}-3
\end{bmatrix}.
\end{equation*}
All $2\times 2$ minors are $3x^{10}+3x^9+2x^8+x^7+2 x^6+2 x^5+x^4+x^3+3 x^2+4x+4$, $2 x^{11}+4x^9+4x^7+x^6+4x^5+3 x^3+x^2+4x+3$, $3 x^{11}+4x^{10}+x^9+x^2+3 x+2$, $4x^{11}+3 x^{10}+3 x^9+2 x^8+4x^7+1$, $4x^{10}+x^9+3 x+2$, and $x^{12}+3 x^9+2 x^3+1$. Then $\mathfrak{d}_2\left(\mathbf{G}\right)=x+2$, which shows that $\mathrm{dim}_{\mathbb{F}_5}\mathcal{C}=11$.
\end{example}

\begin{example}
\label{ex12}
In Example \ref{ex6}, we consider the single generator $\left(\omega^2,\omega\right)$-MT code over $\mathbb{F}_4$ with block lengths $\left(4,8\right)$ generated by
\begin{equation*}  
\begin{split}
\mathbf{g}&=\left( 1 ,  \omega+ \omega x + x^2 + \omega^2 x^3 + \omega^2 x^4 + \omega^2 x^5 + \omega x^6 +x^7 \right).
\end{split}
\end{equation*}
We use Theorem \ref{dimension} to determine the code dimension. We have $\mathrm{dim}_{\mathbb{F}_4}\mathcal{C}=12-\deg\left(\mathfrak{d}_2\left(\mathbf{G}\right)\right),$ where $\mathfrak{d}_2\left(\mathbf{G}\right)$ is the greatest common divisor of all $2\times 2$ minors of  
\begin{equation*}
\mathbf{G}=\begin{bmatrix}
\mathbf{g}\\
\mathrm{diag}\left[x^{m_i}-\lambda_i\right]
\end{bmatrix}=\begin{bmatrix}
1 &  \omega+ \omega x + x^2 + \omega^2 x^3 + \omega^2 x^4 + \omega^2 x^5 + \omega x^6 +x^7  \\
x^{4}-\omega^2 & 0\\
0 & x^{8}-\omega
\end{bmatrix}.
\end{equation*}
All $2\times 2$ minors are $x^{11}+\omega x^{10}+\omega^2 x^9+\omega^2 x^8+\omega x^3+\omega^2 x^2+x+1$, $x^8+\omega$, and $x^{12}+\omega^2 x^8+\omega x^4+1$. Then $\mathfrak{d}_2\left(\mathbf{G}\right)=x^8+\omega$, which shows that $\mathrm{dim}_{\mathbb{F}_4}\mathcal{C}=4$.
\end{example}

\section*{Acknowledgment}
This research was conducted at Université d'Artois, La Faculté Jean Perrin in Lens, and was funded by the Science, Technology \& Innovation Funding Authority (STDF); International Cooperation Grants, project number 49294. Ramy Takieldin would like to express his deepest gratitude to Professor André Leroy for his invaluable guidance throughout this project.

\bibliographystyle{plain}

\end{document}